\documentclass[8pt]{article}
\usepackage{amsmath}
\usepackage{amssymb,amsfonts}
\usepackage{amsthm}
\usepackage{thmtools, thm-restate}
\usepackage{mathrsfs}
\usepackage{comment}
\usepackage{bbm}

\declaretheorem[{style=definition,numberwithin=section}]{definition}
\declaretheorem[{style=definition,sibling=definition}]{theorem}
\declaretheorem[{style=definition,sibling=definition}]{lemma}
\declaretheorem[{style=definition,sibling=definition}]{claim}
\declaretheorem[{style=definition,sibling=definition}]{corollary}
\declaretheorem[{style=definition,sibling=definition}]{proposition}
\declaretheorem[{style=definition,sibling=definition}]{observation}

\declaretheorem[{style=plain     ,sibling=definition}]{remark}
\declaretheorem[{style=definition,name=Question}]{question}

\newcommand{\acc}{AC^0[$p$]}
\newcommand{\acct}{AC^0[\oplus]}
\newcommand{\eps}{\varepsilon}

\title{Represent MOD function by low degree polynomial with unbounded one-sided error}
\author{Chris Beck and Yuan Li}
\date{}

\begin{document}
\maketitle

\begin{abstract}
In this paper, we prove tight lower bounds on the smallest degree of a nonzero polynomial in the ideal generated by $MOD_q$ or $\neg MOD_q$ in the polynomial ring $F_p[x_1, \ldots, x_n]/(x_1^2 = x_1, \ldots, x_n^2 = x_n)$, $p,q$ are coprime, which is called \emph{immunity} over $F_p$. The immunity of
$MOD_q$ is lower bounded by $\lfloor (n+1)/2 \rfloor$, which is achievable when $n$ is a multiple of $2q$; the immunity of $\neg MOD_q$ is exactly
$\lfloor (n+q-1)/q \rfloor$ for every $q$ and $n$.
Our result improves the previous bound $\lfloor \frac{n}{2(q-1)} \rfloor$ by Green.

We observe how immunity over $F_p$ is related to $\acc$ circuit lower bound. For example, if the immunity of $f$ over $F_p$ is lower bounded by $n/2 - o(\sqrt{n})$, and $|1_f| = \Omega(2^n)$, then $f$ requires
$\acc$ circuit of exponential size  to compute.
\end{abstract}

\section{Introduction}

A fundamental task in computer science is to take simple functions like $OR$, $MAJ$, $MOD_q$, etc. and determine how difficult it is to represent them as polynomials over a
field $F_p$. Usually for such questions, it is easy to determine the degree required to exactly represent such a function, but when we ask a variant such as how hard it is to
compute an approximation in low degree it can become quite difficult to get tight results. We were led to the following question by way of proof complexity and circuit complexity (for example, \cite{MA01} proved general hardness criterion for Polynomial Calculus based on immunity).

\begin{question}
What is the smallest degree of a nonzero polynomial in the ideal generated by $MOD_q$ or $\neg MOD_q$ in the polynomial ring $F_p[x_1, \ldots, x_n]/(x_1^2 = x_1, \ldots, x_n^2 = x_n)$?
\end{question}

There are several definitions of the same concept. In \cite{MA01}, this value is called the \emph{immunity} of the $MOD_q$ function, where
the immunity of Boolean function $f:\{0,1\}^n \to \{0,1\}$ over some field $F$ is the minimal degree of some nonzero function in the ideal $\langle f \rangle$,
where $f$ here is viewed as a polynomial in the ring $F[x_1, \ldots, x_n]/(x_1^2 = x_1, \ldots, x_n^2 = x_n)$.
It is also known in the literature as the \emph{weak $p$-degree} \cite{Gre00} of Boolean function $f$ , and can
alternately be defined as the smallest degree of a nontrivial polynomial $g \in \mathbb{Z}[x_1, \ldots, x_n]$ such that
\[ \forall x \in \{0,1\}^n, f(x) = 0 \Rightarrow g(x) \equiv 0 \pmod p \ .\]
Immunity can also be thought of as measuring how expensive it is to compute a function with unbounded one-sided error on $1$ by polynomial.

In cryptography community, \emph{algebraic immunity} of a Boolean function $f$ is defined to be the smaller one between immunity of $f$
and immunity of $1-f$ over $F_2$. In this paper, we will call it ``two-sided immunity''
to differentiate with the (one-sided) immunity. In cryptography, two-sided immunity is a criterion of the security
of Boolean functions used in stream cipher system \cite{CDG06}.


For purpose of our current research, we hoped for a tight lower bound of the immunity of functions $MOD_q$ and $\neg MOD_q$.
However, with not much trouble we were able to show an improved lower bound $n/2$ of $MOD_q$ function, which relies only on the kind of analysis that appears in Razborov-Smolensky.
The lower bound $n/2$ of $MOD_q$ turns out to be tight when $n$ is a multiple of $2q$.  For $\neg MOD_q$, we prove
an exact result of immunity over $F_p$, which is $\lfloor (n+q-1)/q \rfloor$ (independent of $p$), based a symmetrization technique, which was
used by Feng and Liu \cite{LF07} in the case of Boolean functions.

Our result improves Green's lower bound $\lfloor \frac{n}{2(q+1)} \rfloor$ \cite{Gre00}, which uses complex Fourier technique. Green's lower bound
improves the results of Barrington, Beigel and Rudich \cite{BBR94} and Tsai \cite{Tsai93}, which proved $\Omega(n)$ lower bound holds for
slow growing $p$.

The paper is organized as follows. In Section 2, we show that weak mod-$m$ degree can be reduced to the case of weak mod-$p$ degree, where $p$ is a prime factor
of $m$, and thus we only need to consider weak mod-$p$ degree, that is the immunity over $F_p$. In Section 3, we prove the $n/2$ lower bound for $MOD_q$. In Section 4, we present the symmetrization
technique in an ideal generated by a symmetric function. In Section 5, we prove an exact result on the immunity of $\neg MOD_q$ based on the symmetrization technique.
Moreover, we give two proofs, and in the second proof, we prove a more general lemma about the rank of a submatrix of the tensor product of matrices satisfying
certain conditions, which might be interesting on its own. In Section 6, we prove a lower bound on the degree of symmetric functions in the ideal generated
by $MOD_q$, which is close to optimal when $n+1$ is a power of $p$ (or slightly larger than a power of $p$), by a restriction technique. In Section 7, we show
the connection between immunity over $F_p$ and $\acc$ circuit lower bound.

\section{Composite Modulus}

For the proof complexity and circuit complexity application we had in mind, we only actually care about prime characteristic.
However, it is natural to try to generalize the improvement to composite characteristic as well.
In fact, we can use a trick similar to what Green did, and reduce the composite case to the prime case.
For extra clarity, we adopt the language which appears in his paper \cite{Gre00}.

\begin{lemma}
The weak mod-$m$ degree of any Boolean function $f$ equals the minimum of weak mod-$p$ degree of $f$,
where $p$ ranges all prime factors of $m$.
\end{lemma}
\begin{proof} One direction is easy, that is, the weak mod-$m$ degree is not greater than the minimum
of the weak mod-$p$ degree. Suppose the minimum weak mod-$p$ degree of $f$, where $p$ ranges
all prime factors of $m$, is $d$. That is, there exists a nonzero (multilinear) polynomial $g \in \mathbb{Z}[x_1, \ldots, x_n]$
of degree $d$ such that
$$
f(x) = 0 \Rightarrow g(x) \equiv 0 \pmod p, \forall x.
$$
Then, we claim that $m/p g(x)$ weakly represent $f$ mod-$m$. Because for all $x$ with $f(x) = 0$,
$g(x) \equiv 0 \pmod p$ implies $m/p g(x) \equiv 0 \pmod m$. And $m/p g(x) \pmod m$ is nonzero because
there exists $x \in \{0,1\}^n$ such that $g(x) \not\equiv 0 \pmod p$, which implies $m/p g(x) \not\equiv 0 \pmod m$.

For the other direction, we need to show the weak mod-$m$ degree of $f$ is not less than the minimum
of the weak mod-$p$ degree. Suppose $g \in \mathbb{Z}[x_1, \ldots, x_n]$ is some multilinear polynomial
weakly represent $f$ with minimum degree $d$. We shall prove that there exists some $g' \in \mathbb{Z}[x_1, \ldots, x_n]$
weakly represent $f$ mod $p$ with degree $\le d$.

Let $x$ be any input such that $g(x)$ is nonzero modulo $m$.
By the Chinese Remainder Theorem, for some maximal prime power $q$ of $m$, $g(x)$ is nonzero modulo $q$,
so $g$ is a nonzero polynomial modulo $q$.

Now suppose that $q = p^e$. If $g$ is a nonzero polynomial modulo $p$ as well, then we are done by letting $g' = g$.
If $g$ as a function is zero modulo $p$ but not modulo $q$, then it is easy to see that every coefficient of $f$ must be zero modulo $p$;
if not, take a monomial $S$ such that for every $T \subsetneq S$, the coefficient of $T$ is zero, then the input such that $x_i = 1$ iff $i \in S$ must have nonzero value modulo $p$.
Thus, if $g$ is zero as a function modulo $p$, but not modulo $q$, its coefficients are all divisible by $p$, and the integer polynomial $g/p$ is nonzero modulo $q/p$.
By iterating this, eventually we obtain a divisor $g'$ of $g$ which is nonzero modulo $p$, and hence $g'$ has degree not greater than that of $g$.
\end{proof}

\section{Lower Bound for $\text{MOD}_q$}

Consider the following quotient of the polynomial ring, $R := F_p [ x_1, \ldots, x_n ] / (x_1^2 = x_1, \ldots , x_n^2 = x_n)$, sometimes called the Razborov-Smolensky ring.
Each element of $R$ has a unique multilinear polynomial representative, and generally we identify an element of $R$ with this representative.
Each polynomial also determines a map from $\{0,1\}^n \to F_p$ by evaluation, and in fact this induces an isomorphism of $F_p$-algebras between $R$ and the algebra of functions $\{0,1\}^n \to F_p$.
So we also will often identify an element of $R$ with function it computes on boolean inputs.

Sometimes authors define the $MOD_q$ function slightly differently in different contexts, and here we will focus on this one first:

\begin{definition}
Let $\chi_q$ denote the element of $R$ defined by
 \begin{equation}
 \chi_q(x_1, \ldots , x_n) := \left\{ \begin{array}{lc} 1 & \textrm{ if } $q$ \textrm{ divides }|\vec{x}| \\ 0 & \textrm{ otherwise }\\ \end{array} \right. \ .
 \end{equation}
\end{definition}

Then, by definition, the immunity/weak $p$-degree of $\chi_q$ is the smallest degree of a nontrivial element of the ideal generated by $\chi_q$ in $R$.

\begin{observation}
$f \in \left< \chi_q \right>$ iff $ f = f \cdot \chi_q$.
\end{observation}
\begin{proof}
By definition, $f \in \left< \chi_q \right>$ if $f = g \cdot \chi_q$ for some $g$, so $(\leftarrow)$ holds. Now lets do $(\rightarrow)$. Since $\chi_q$ is zero-one valued, $\chi_q^2 = \chi_q$, so $f \cdot \chi_q = g \cdot \chi_q^2 = g \cdot \chi_q = f$, so $(\rightarrow)$ holds as well.
\end{proof}

Following the general Razborov-Smolensky methodology, let $\omega$ denote a primitive $q$'th root of unity found in some large enough extension field of $F_p$ (if $F_p$ did not already contain $\omega$, observe that $\left< \chi_q \right>$ contains only more polynomials when we work over a larger field). Note that this does not require that $q$ be a prime. Define new variables $y_i := 1 + (\omega -1) x_i$. Then the $y_i$ are elements of $R$, but also $x_i$ is determined by $y_i$ so if we like for any function $f(x) \in R$, we can think of it as a function $f(y) : \{1, \omega\}^n \to F_p$. Of course it has a unique multilinear representation in the variables $y_i$ as well. While the coefficients might look different, its degree in this representation must be the same, because the degree of a polynomial cannot increase under a linear transformation of the variables, and our linear transformation is invertible.

We will also introduce variables $y'_i := 1 + (\omega^{-1} - 1) x_i$, and by the same reasoning, we know that for any $f \in R$, its degree as represented in the $x_i$, $y_i$, or $y'_i$ is the same. Note also that $y_i \cdot y'_i = 1$ as elements of $R$.

\begin{observation}
$f \in \left< \chi_q \right> $ iff $ f = f \cdot \prod_i y'_i$.
\end{observation}
\begin{proof}
Think of $\prod_i y'_i$ as a function in the $x$-variables. Because $\omega$ is a $q$'th root of unity, $\prod_i y'_i \neq 1$ if and only if $\chi_q = 0$. Thus, $\chi_q \cdot \left(\prod_i y'_i - 1 \right) = 0$. Therefore, for any $f \in \left< \chi_q \right>$, \[ f \cdot \left(\prod_i y'_i -1\right) = f \cdot \chi_q \cdot \left(\prod_i y'_i - 1\right) = 0 \ ,\] so $f \cdot \prod_i y'_i = f$.
\end{proof}

Now we use this to prove the main result.

\begin{theorem}
If $f \in \left< \chi_q \right>$, then $f =0$ or  $f$ has degree $\geq n/2$.
\end{theorem}
\begin{proof}
Suppose not. Consider $f$'s representation as a polynomial in the $y_i$, \[ f = \sum_S c_S \prod_{i \in S} y_i \ .\]

For any monomial $S$, we have that \[ \prod_{i \in S} y_i \cdot \prod_i y'_i = \prod_{i \in \overline{S}} y'_i \ .\] Since $f = f \cdot \prod_i y'_i$, we deduce that $f$'s representation as a $y'_i$ polynomial is
\[ f = \sum_S c_S \prod_{i \in \overline{S}} y'_i \ .\]
If the polynomial $f(y)$ is nonzero and has degree less than $n/2$, then this polynomial representation of $f(y')$ has at least one nonzero monomial of degree strictly larger than $n/2$, and so has degree greater than $n/2$. But this is a contradiction, since as we saw before, the degree of the polynomials $f(y)$ and $f(y')$ must be the same, as they are linear transformations of one another.
\end{proof}

Note that nowhere did we assume that $q$ was not composite, only that it is coprime with $p$, which is sufficient to find a $q$'th root of unity in a large enough extension of $F_p$.

The idea in the above proof also can be used to show an upper bound of the immunity of $\neg \chi_q$. Again, $\omega$ is a $q$th root
of unity, and $y_i = (\omega - 1)x_i + 1$ and $y_i' = (\omega^{-1} - 1)x_i + 1$, which is the inverse of $y_i$. It's easy to see that
$$
\prod_{i \le n/2} y_i = \prod_{i > n/2} y_i'
$$
holds for all $x$ with $|x| \equiv 0 \pmod q$, since $1 = \prod_{i \le n} y_i = \prod_{i \le n/2} y_i (\prod_{i > n/2} y_i')^{-1}$, which implies
$$
\prod_{i \le n/2} y_i - \prod_{i > n/2} y_i' \in \langle \neg \chi_q \rangle.
$$
Thus, the immunity of $\neg \chi_q$ is upper bounded by $\lceil n/2 \rceil$.

\vspace{0.2cm}
The tightness of the lower bound $n/2$ is shown by the following example.
Let  $n$ be even and $n/2 \equiv 0 \pmod q$, and let
$$
g = \prod_{i=0}^{n/2} (x_{2i-1} - x_{2i}).
$$
It's easy to see $g \in \langle \chi_q \rangle$, because $g(x) = 0$ for all $x$ with $|x| \not= n/2$, and
thus, $g(x) = 0$ for all $x$ with $|x| \equiv 0 \pmod q$.

\section{Symmetrization}

One key ingredient of our improved lower bound for $\neg \chi_q$ is the fact that we can symmetrize any function in a symmetric ideal, where symmetric ideal is an ideal generated by a symmetric function.
If the characteristic of the field is zero, this is trivial, for we can summing over all permutations of some given function to obtain a symmetric one with algebraic degree non-increasing.
When working over  finite field, this averaging technique does not work because we may get a zero function.

However, we could still symmetrize an annihilator to some simple form, as the following lemma says.
The following lemma is proved by Feng and Liu in the case of Boolean functions, that is, $F=F_2$ \cite{LF07}. For the ring
$F[x_1, \ldots, x_n]/ (x_1^2 = x_1, \ldots, x_n^2 = x_n)$, the proof is almost the same.
The idea is to symmetrize step by step in order to avoid getting a zero function in contrast to summing over all permutations in the case
of characteristic zero.

\begin{lemma}
\label{lem:sym_ann_gen}If $f \in F[x_1, \ldots, x_n]/ (x_1^2 = x_1, \ldots, x_n^2 = x_n)$ is a symmetric function, there is a lowest degree $g$ in $\langle f \rangle$ of the following form
\begin{equation}
g = g' \prod_{i = 1}^\ell (x_{2i-1} - x_{2i}),
\end{equation}
where $g'$ is a symmetric function on variables $x_{2\ell+1}, \ldots, x_n$.
\end{lemma}
\begin{proof} Prove by construction. Let $g$ be a function in $\langle f \rangle$ with lowest degree. If $g$ is symmetric, then we are done. Thus assume $g$ is not symmetric.
Since the symmetric group $S_n$ is generated by all transpositions $(i, j)$, $1 \le i < j \le n$, the assumption that $g$ is not symmetric implies there exists some
transposition $\pi = (i, j)$ such that $\pi(g) \not= g$. Let
$$g' = g - \pi(g) \not= 0.$$

In fact, $g' = (x_i - x_j) h$, where $h$ is a symmetric function on $\{x_1, \ldots, x_n\} \setminus \{x_i, x_j\}$. To see this, write
$
g = g_0 + g_1 x_i + g_2 x_j + g_3 x_i x_j,
$
where $g_0, g_1, g_2, g_3$ are functions on $\{x_1, \ldots, x_n\} \setminus \{x_i, x_j\}$.
And thus
$
\pi(g) = g_0 + g_1 x_j + g_2 x_i + g_3 x_i x_j,
$
which implies
$$
g - \pi(g) = (x_i - x_j) (g_1 - g_2)
$$
Repeat this procedure on $h = g_1 - g_2$ until one gets a symmetric function. Finally, we find a function $g$ in ideal $\langle f \rangle$ with the following form
$$
g = g' \prod_{i = 1}^\ell (x_{t_{2i-1}} - x_{t_{2i}}),
$$
indexes $t_1, t_2, \ldots, t_{2\ell}$ can take $1, 2, \ldots, 2\ell$ because we could apply a permutation $\pi$ to $g$ which sends $t_i$ to $i$, which is in
the ideal $\langle \pi(f) \rangle = \langle f \rangle$ for $f$ is invariant under all permutations.
\end{proof}

The above lemma has the following consequence. In order to lower bound the degree of nonzero functions in some symmetric ideal  $\langle f \rangle$ in $R$,
we only need to consider all functions of the form $g = g' \prod_{i = 1}^\ell (x_{2i-1} - x_{2i}),$ where $g'$ is symmetric on variables $x_{2\ell+1}, \ldots, x_n$.
The fact that $f(x) =  0 \Rightarrow g(x) = g' \prod_{i = 1}^\ell (x_{2i-1} - x_{2i}) = 0$ is equivalent to $f|_\rho(x) = 0 \Rightarrow g'(x) = 0$ where
$\rho$ is the restriction setting $x_{1} = x_{2} = \ldots = x_{2\ell-1} = 0$ and $x_{2}=x_4=\ldots=x_{2\ell}=1$, that is, $g'$ is in the ideal $\langle f|_\rho \rangle$.
Therefore, we have the following corollary.

\begin{corollary}
\label{cor:sym}
Let $f \in F[x_1, \ldots, x_n]/ (x_1^2 = x_1, \ldots, x_n^2 = x_n)$ be a symmetric function. The lowest degree of a nonzero function in $\langle f \rangle$ equals
the minimum degree of $\deg(g) + \ell$, where $ g \in \langle f|_{\rho} \rangle$ and $\rho$ ranges over all restrictions setting $x_{1} = x_{2} = \ldots = x_{2\ell-1} = 0$ and $x_{2}=x_4=\ldots=x_{2\ell}=1$,
$0 \le \ell \le n/2$.
\end{corollary}

\section{Lower Bound for $\neg \text{MOD}_q$}

By Corollary \ref{cor:sym}, in order to prove symmetric $f$ has immunity not less less than $d$, it's equivalent to prove any nonzero symmetric
function in $\langle f|_{\rho_i} \rangle$ has degree not less than $d-i$, for $i = 0, 1, \ldots, \min\{\lfloor n/2 \rfloor, d\}$,
where restriction $\rho_i$ sets $x_1, x_3, \ldots, x_{2i-1}$ to $1$, and $x_2, x_4, \ldots, x_{2i}$ to $0$.

It's easily checked that if the truth value table of symmetric function $f$ is
$$
v_f = (v_f(0), v_f(1), \ldots, v_f(n)) \in F_2^{n+1},
$$
then the truth value table of ${f|_{\rho_i}}$ is
$$
v_{f|_{\rho_i}} = (v_f(i), v_f(i+1), \ldots, v_f(n-i)) \in F_2^{n+1-2i}.
$$

Assume function $g$ is a symmetric function in $\langle f \rangle$ of degree less than $d$, and we can write
$
g = \sum_{i<d} c_i \sigma_i,
$
where $\sigma_i$ is the elementary symmetric polynomial of degree $i$.
For convenience, we define function $\psi_d : \mathbb{N} \to F_2^{d}$ by
$$
\psi_d(i) = (\binom{i}{0}, \binom{i}{1}, \ldots, \binom{i}{d-1} ) \in  F_p^{d},
$$
which is the evaluation $\sigma_0, \sigma_1, \ldots, \sigma_{d-1}$ at value $i$.
The fact $g \in \langle f \rangle$ implies $g(w) = 0$ for all $w$ with $v_f(w) = 0$, that is
$$
\begin{pmatrix}
\psi_d(i_1) \\
\psi_d(i_2) \\
\vdots \\
\psi_d(i_t)
\end{pmatrix}_{t \times d}
\begin{pmatrix}
c_0 \\
c_1 \\
\vdots \\
c_d
\end{pmatrix} = 0,
$$
where $v_f(i_1) = \ldots = v_f(i_t) = 0$. Therefore, $\langle \neg \chi_q \rangle$ has nonzero symmetric function of degree less than $d$ if and only if the rank of
$\{\psi_d(w) : \chi_q(w) = 1 \}$ is smaller than $d$. It turns out the rank of $\{\psi_d(w) : \chi_q(w) = 1 \}$
is always full (equals the number of vectors).

The lower bound of immunity of $\neg \chi_q$ follows from the following lemma. We will present two proofs of the following lemma, and the first one
is much simpler. However, we are reluctant to discard the second one since it has a byproduct as we will later see.

\begin{lemma}
\label{lem:basis}
 Fix a prime $p$. Let integers $a \ge 0$ and $d > 0$, and $q$ is coprime to $p$. Vectors
$$
\psi_d(a), \psi_d(a+q), \ldots, \psi_d(a+(d-1)q) \in F_p^{d}
$$
is a basis $F_p^{d}$.
\end{lemma}
\begin{proof} It suffices to prove the determinant of $\psi_d(a), \psi_d(a+q), \ldots, \psi_d(a+(d-1)q)$ is nonzero, which turns out to
have a simple closed form.

For convenience, let $a_i = a + iq$.
\begin{eqnarray*}
& &
\det
\begin{pmatrix}
 \binom{a_0}{0} & \binom{a_0}{1} & \cdots & \binom{a_0}{d-1}\\
\binom{a_1}{0} & \binom{a_1}{1} & \cdots & \binom{a_1}{d-1}\\
\vdots & \vdots & \ddots & \vdots \\
\binom{a_{d-1}}{0} & \binom{a_{d-1}}{1} & \cdots & \binom{a_{d-1}}{d-1}
\end{pmatrix} \label{equ:equ1} \\
& = & \frac{1}{{\prod_{k=1}^{d-1} k! }} \det
\begin{pmatrix}
 1 & a_0 & a_0(a_0 - 1) & \cdots & a_0(a_0-1)\cdots(a_0-d-2)\\
 1 & a_1 & a_1(a_1 - 1) & \cdots & a_1(a_1-1)\cdots(a_1-d-2)\\
\vdots & \vdots & \vdots & \ddots & \vdots \\
 1 & a_{d-1} & a_{d-1}(a_{d-1} - 1) & \cdots & a_{d-1}(a_{d-1}-1)\cdots(a_{d-1}-d-2)
\end{pmatrix} \\
& = & \frac{1}{{\prod_{k=1}^{d-1} k!}}
\det
\begin{pmatrix}
 1 & a_0 & a_0^2 & \cdots & a_0^{d-1}\\
 1 & a_1 & a_1^2 & \cdots & a_1^{d-1}\\
\vdots & \vdots & \vdots & \ddots & \vdots \\
 1 & a_{d-1} & a_{d-1}^2 & \cdots & a_{d-1}^{d-1}
\end{pmatrix}\\
& = &  \frac{\prod_{0 \le i < j \le d-1} (a_j - a_i)}{{\prod_{k=1}^{d-1} k! }} \\
& = &  \frac{\prod_{0 \le i < j \le d-1} q(j-i)}{{\prod_{k=1}^{d-1} k! }}\\
& = & q^{d(d-1)/2},
\end{eqnarray*}
which is nonzero for $q$ is coprime to $p$. In the above calculation, the first step is by the definition of binomial coefficients;
the second step is by adding column $i$ to column $i+1$ for $i = 1, 2, \ldots, d-1$; and the third step is by Vandermonde determinant
formula.
\end{proof}

Now we can calculate the immunity of $\neg \chi_q$, which is defined to be the minimal degree of a nonzero function in the ideal $\langle \neg \chi_q \rangle$.

\begin{theorem}
Let $p$ be a prime, and $q \ge 2$ an integer coprime to $p$. The immunity of $\neg \chi_q$ over $F_p$ is $\lfloor \frac{n+q-1}{q} \rfloor$, which is
independent of $p$.
\end{theorem}
\begin{proof}
By Lemma \ref{lem:basis}, the minimal degree of nonzero symmetric function in $\langle \neg \chi_q \rangle$ is the weight
of $\chi_q$, which is $\lfloor \frac{n}{q} \rfloor+1$;
the minimal degree of nonzero symmetric function in $\langle \neg \chi_q|_{\rho_1} \rangle$ is the weight of $\chi_q|_{\rho_1}$, which is
 $\lfloor \frac{n-1}{q} \rfloor$;
...; the minimal degree of nonzero symmetric function in $\langle \neg \chi_q|_{\rho_i} \rangle$ is $\lfloor \frac{n-i}{q} \rfloor - \lfloor \frac{i-1}{q} \rfloor$.

Therefore, the immunity of of $\neg \chi_q$ is
\begin{eqnarray*}
& & \min\{\lfloor \frac{n}{q} \rfloor + 1, \lfloor \frac{n-i}{q} \rfloor - \lfloor \frac{i-1}{q} \rfloor + i : i = 1, \ldots, \lfloor n/2 \rfloor \}, \\
& = & \min\{\lfloor \frac{n+q-1}{q} \rfloor \},
\end{eqnarray*}
which is easy to check.
\end{proof}

Now, let's present the second proof Lemma \ref{lem:basis} by proving a more general result about the rank of tensor product
of matrices.

\begin{definition}
\label{def:nondeg}
 Call $A \in M_{n \times n}(F)$ \emph{strong nondegenerate matrix} if for any $1 \le t \le n$ and
$1 \le i_1 < \ldots < i_t \le n$, submatrix $M(i_1, \ldots, i_t; 1, \ldots, t)$ always has full rank $t$.

Call $A \in M_{n \times n}(F)$ \emph{weak nondegenerate matrix} if for any $1 \le t \le n$ and any integer $a$,
and any $q$ coprime to $n$, submatrix $M(a, a+q, \ldots, a+(t-1)q; 1, \ldots, t)$ always has full rank $t$,
where the row indexes are computed mod the size of matrix $A$.
\end{definition}

By the definition, a strong nondegenerate matrix is always weak nondegenerate. The following theorem says we can construct many weak nondegenerate
matrices by taking tensors products of strong nondegenerate ones.

\begin{theorem} The tensor product of strong nondegenerate matrices is weak nondegenerate.
\end{theorem}
\begin{proof} Suppose $A_1, \ldots, A_m$ are strong nondegenerate matrices, and $q > 0$ is coprime to the size
of each $A_i$. We need to prove matrix
$$
(A_1 \otimes \ldots \otimes A_m)(a, a+q, \ldots, a+(d-1)q; 1, \ldots, d)
$$
has full rank $d$. Let $\ell$ be the size of $B$, and let $A_1$ be $p \times p$ matrix, and thus $q$ is coprime to
both $p$ and $\ell$.

Prove by induction on $m$. For the basis $m=1$, the conclusion is trivial by the definition of nondegenerate matrix. Let's assume it's true for
$m-1$, and prove it for $m$. Let $B = A_2 \otimes \ldots \otimes A_m$. Recalling the definition of tensor product,
\begin{equation}
\label{equ:tensor}
A_1 \otimes B = \begin{pmatrix}
a_{11} B & a_{12}B & \ldots & a_{1p} B \\
a_{21} B & a_{22}B & \ldots & a_{2p} B \\
\vdots & \vdots & \ddots & \vdots\\
a_{p1} B & a_{p2}B & \ldots & a_{pp} B \\
\end{pmatrix}
\end{equation}
Let $d = \lfloor d/\ell \rfloor \ell + d'$.


\textbf{Case 1: $d \le \ell$.} By the definition the non-degenerate matrix (Definition \ref{def:nondeg}), $a_{i1}$, $i=1,\ldots, p$, are nonzero in the field $F$.
Thus,
\begin{eqnarray*}
 & & \langle A(a;1, \ldots, d), A(a+q;1, \ldots, d), \ldots, A(a+(d-1)q;1, \ldots, d) \rangle \\
& = & \langle B(a;1, \ldots, d), B(a+q;1, \ldots, d), \ldots, B(a+(d-1)q;1, \ldots, d) \rangle,
\end{eqnarray*}
which has full rank by induction
hypothesis on $m-1$.

\textbf{Case 2: $d > \ell$ and $d' = 0$.} Since $q$ is coprime to $p\ell$, $d = \lfloor d/\ell \rfloor \ell$ numbers $a, a+q, \ldots, a+(d-1)q$ runs
over $\{0, 1, \ldots, \ell-1\}$ for exactly $t = d/\ell$ times, which implies for any $j \in \{0, 1, \ldots, \ell-1\}$, there exists
$t$ distinct numbers $i_1, i_2, \ldots, i_t \in \{a, a+q, \ldots, a+(d-1)q\}$ which is congruent to $j$ mod $\ell$.

For convenience, let $B(i)$ denotes the $i$th row of $B$, and let
$$B^{(c)}(i) = \underbrace{(0, \ldots, 0)}_{(c-1)\ell} \oplus B(i) \oplus \underbrace{(0, \ldots, 0)}_{(t-c-1)\ell},$$
where $c = 1, \ldots, t$. Let $F^d = S_1 \oplus \ldots \oplus S_t$, where $S_i$ is the subspace of $F^d$ of dimension
$\ell$, generated by $e_{(i-1)\ell+1}, \ldots, e_{i\ell}$.

By definition of tensor product \eqref{equ:tensor},
\begin{eqnarray*}
& A(i_1; 1, \ldots, d) = a_{i'_1,1} B^{(1)}(j) + a_{i'_1, 2} B^{(2)}(j) + \ldots + a_{i'_1, t} B^{(t)}(j) \\
& A(i_2; 1, \ldots, d) = a_{i'_2,1} B^{(1)}(j) + a_{i'_2, 2} B^{(2)}(j) + \ldots + a_{i'_2, t} B^{(t)}(j) \\
& \ldots  \ldots \ldots \\
& A(i_t; 1, \ldots, d) = a_{i'_t,1} B^{(1)}(j) + a_{i'_t, 2} B^{(2)}(j) + \ldots + a_{i'_t, t} B^{(t)}(j),
\end{eqnarray*}
where $i'_k = \lfloor i_k / \ell \rfloor $. Since matrix $A_1$ is non-degenerate, the coefficient matrix $(a_{i'_j, k})_{j,k=1,\ldots, t}$ is
invertible, which implies
\begin{eqnarray*}
\langle B^{(1)}(j), \ldots, B^{(t)}(j) \rangle & \subseteq & \langle A(i_1; 1, \ldots, d), \ldots, A(i_t; 1, \ldots, d) \rangle \\
& \subseteq & \langle A(a; 1, \ldots, d), \ldots, A(a+(d-1)q; 1, \ldots, d) \rangle
\end{eqnarray*}
Since $j \in \{0, \ldots, \ell-1\}$ is arbitrary, we have $B^{(c)}(0), \ldots, B^{(c)}(\ell-1), c = 1, \ldots, t$, is in the linear span of
$A(a;1,\ldots,d), \ldots, A(a+(d-1)q;1,\ldots,d)$. By induction hypothesis, $B^{(c)}(0), \ldots, B^{(c)}(\ell-1)$ is a basis of subspace $S_c$ of dimension $\ell$ in $F^d$;
since $F^d$ is the direct sum of $S_1, \ldots, S_t$, we complete the proof of this case.

\textbf{Case 3: $d > \ell$ and $d' > 0$.} Since $q$ and $p\ell$ are coprime, $d-d'$ numbers $a+d'q, \ldots, a+(d-1)q$ runs
over $\{0, 1, \ldots, \ell-1\}$ for exactly $t = \lfloor d/\ell \rfloor$ times, and the extra $d'$ numbers $a, \ldots, a+(d'-1)q$
 numbers are distinct mod $\ell$. This implies for any $j \in \{a, \ldots, a+(d'-1)q\}$, there exists
$t +1$ distinct numbers $i_1, i_2, \ldots, i_t \in \{a, a+q, \ldots, a+(d-1)q\}$ which is congruent to $j$ mod $\ell$.

Similar to Case 2, let $B(i)$ denotes the $i$th row of $B$, and let
$$B^{(c)}(i) = \underbrace{(0, \ldots, 0)}_{(c-1)\ell} \oplus B(i) \oplus \underbrace{(0, \ldots, 0)}_{d - c\ell},$$
where $c = 1, \ldots, t$. However, for $c = t+1$, let
$$B^{(t+1)}(i) = \underbrace{(0, \ldots, 0)}_{t\ell} \oplus B(i;1, \ldots, d').$$
Again, by definition of tensor product \eqref{equ:tensor},
\begin{eqnarray*}
& A(i_1; 1, \ldots, d) = a_{i'_1,1} B^{(1)}(j) + a_{i'_1, 2} B^{(2)}(j) + \ldots + a_{i'_1, t+1} B^{(t+1)}(j) \\
& A(i_2; 1, \ldots, d) = a_{i'_2,1} B^{(1)}(j) + a_{i'_2, 2} B^{(2)}(j) + \ldots + a_{i'_2, t+1} B^{(t+1)}(j) \\
& \ldots  \ldots \ldots \\
& A(i_{t+1}; 1, \ldots, d) = a_{i'_{t+1},1} B^{(1)}(j) + a_{i'_{t+1}, 2} B^{(2)}(j) + \ldots + a_{i'_{t+1}, {t+1}} B^{(t+1)}(j),
\end{eqnarray*}
where $i'_k = \lfloor i_k / \ell \rfloor $. Since matrix $A_1$ is non-degenerate, the coefficient matrix $(a_{i'_j, k})_{j,k=1,\ldots, t+1}$ is
invertible, which implies
$$
\langle B^{(1)}(j), \ldots, B^{(t)}(j) \rangle \subseteq \langle A(i_1; 1, \ldots, d), \ldots, A(i_{t+1}; 1, \ldots, d) \rangle.
$$
Since $j \in  \{a, \ldots, a+(d'-1)q\}$ is arbitrary, we conclude $B^{(t+1)}(a), \ldots, B^{(t+1)}(a+(d'-1)q)$, is in the linear span of
$A(a;1,\ldots,d), \ldots, A(a+(d-1)q;1,\ldots,d)$. By induction hypothesis, $B^{(t+1)}(a), \ldots, B^{(t+1)}(a+(d'-1)q)$ is a basis of $S_{t+1}$. After mod
out $S_{t+1}$ from $F^d$, and repeat the argument as in Case 2, the proof is complete.
\end{proof}

Lemma \ref{lem:basis} follows from the above theorem by taking
$
A = (\binom{i}{j})_{i, j = 0,\ldots, p -1}
$
and thus $\psi_d(i) = (A \otimes A \otimes \cdots \otimes A)(i; 1, \ldots, d)$ by Lucas formula. The fact that $A$ is a non-degenerate matrix can be shown by
computing its determinant as in the proof of Lemma \ref{lem:basis}.

\section{Lower Bound for Symmetric Functions in $\langle \chi_q \rangle$}

By the result in Section 4, to lower bound the immunity of $\chi_q$, it's equivalent to lower bound the degree of \emph{symmetric} functions
in the ideal $\langle \chi_q \rangle, \langle \chi_q|_{\rho_1} \rangle, \ldots$, where $\rho_i$ is the restriction sending $x_{2j-1}$ to $0$ and $x_{2j}$ to $1$ for
$j = 1, 2, \ldots, i$. When restricting our attention to only symmetric functions, it becomes much easier to deal with.

In this section, we will lower bound the degree of symmetric functions in $\langle \chi_q \rangle$, and the result here is not strong enough to
prove $\lceil n/2 \rceil$ lower bound for every $n$. However, in some special cases, such as $n+1$ is a power of $p$, we will prove better lower
bound on the degree of nonzero symmetric functions in $\langle \chi_q \rangle$ which is close to optimal.

Let $f : \{0, 1\}^n \to F$ be a symmetric function in $R$, and let $v_f: \{0, 1, \ldots, n\} \to F$ be its value vector, i.e., $v_f(|x|) = f(x)$. It's clear that
any symmetric function in $R$ can be written as a linear combination of elementary symmetric polynomials $\sigma_0, \ldots, \sigma_n$, that is,
\begin{equation}
 \label{equ:sym_basis}
f(x) = \sum_{i = 0}^n c_f(i) \sigma_i(x),
\end{equation}
where $c_f = (c_f(0), \ldots, c_f(n)) \in F^{n+1}$ is the coefficients of $f$ in the above form.
Given $c_f$, the value of $v_f$ is determined by
\begin{equation}
 \label{equ:vfromc}
v_f(i) = \sum_{j = 0}^i \binom{i}{j} c_f(j).
\end{equation}
By a special case of Mobius inversion on the Boolean lattice, $c_f$ can be written in $v_f$ as follows,
\begin{equation}
\label{equ:cfromv}
c_f(i) = \sum_{j = 0}^i (-1)^{i+j} \binom{i}{j} v_f(j).
\end{equation}
The following proposition is an immediate consequence from equations \eqref{equ:vfromc} and \eqref{equ:cfromv}.

\begin{proposition}
\label{prop:reflex}
 There is a symmetric function in $R$ of degree less than $d$ supported only on points of hamming weight in $S \subseteq \{0,1,\ldots, n\}$
if and only if there is a symmetric function in $R$ supported only on monomials of weight in $S$ which takes value zero on every input
point of hamming weight not less than $d$.
\end{proposition}

By the above proposition, the following lemma implies the lower bound of the degree of symmetric functions in the ideal $\langle \chi_q \rangle$,
when $n+1$ is a power of $p$.

\begin{lemma}
\label{lem:restrict}
 Let $f \in R$ be a nonzero symmetric function supported only on monomials of weight in $S_a = \{a, a+q, a+2q, \ldots, \} \subseteq \{0, 1, \ldots, N = p^n - 1\}$, which
takes value zero on every input point of hamming weight not less than $d$. Then,
$$
d \ge N ( 1 - \frac{1}{p^\ell} ),
$$
where $\ell = \lfloor \log_p (q-1) \rfloor$.
\end{lemma}
\begin{proof}
If $f$ is symmetric supported only on monomials of weight in $S_a$, and $w$ is an integer variable representing the weight $|x|$ of $x$, we can express $f : \{0, 1, \ldots, N = p^n-1\} \to \mathbb{F}_p$ as
\[ f(w) = \sum_{k} c_k \binom{w}{k} .\]

Now we employ Lucas' Theorem, in the mod $p$ case.
\begin{claim}
\[ \binom{w}{k} \equiv \prod_{i=0}^{n-1} \binom{w_i}{k_i} \pmod p \]
where $w_i, k_i$ are the $i$'th bits in the $p$-adic representation of $w,k$ respectively.
\end{claim}

It is easy to see that $\binom{w_i}{0} =1, \binom{w_i}{1} = w_i, \ldots, \binom{w_i}{j} = w_i (w_i - 1) \ldots (w_i - j + 1) / j!, \ldots, \binom{w_i}{p-1} = w_i (w_i - 1) \ldots (w_i - p + 2) / (p-1)!$
which are linearly independent in the polynomial ring $F_p[w_i]$. Let's veiw $\binom{w}{k}$ as a polynomial of $w_0, w_1, \ldots, w_{n-1}$. From the linear independence of $\binom{w_i}{0}, \ldots,
\binom{w_i}{p-1}$, we claim terms $\binom{w}{0}, \binom{w}{1}, \ldots, \binom{w}{p^{n} - 1}$ are linearly independent as polynomials in $\mathbb{F_p}[w_0, \ldots, w_{n-1}]$.

Let's write
\[ f(w) = \sum_{k} c_k \binom{w}{k}  = \sum_{k} c_k \prod_{i=0}^{n-1} \binom{w_i}{k_i},\]
and view it as a polynomial in $\mathbb{F_p}[w_0, \ldots, w_{n-1}]$. We will show that if $c_k = 0$ except when $k \in S_a$, then $f$ takes a nonzero value of
large hamming weight as a function $\{0, 1, \ldots, N\} \to \mathbb{F}_2$.
To achieve this, fix a parameter $\ell$. We will hit $f$ with a restriction which sets the $\ell$ highest order bits of input $w$ to $p-1$ --
if we can prove that the restricted polynomial is nonzero, it implies there is a nonzero point of value at least $(p^\ell - 1)p^{n-\ell} = N(1 - p^{-\ell})$. Thus we
would like to do this with $\ell$ as large as possible. Let $\rho$ denote this restriction.

What happens when we restrict a term (here, term is specifically refer to a multiple of $\binom{w}{0}, \binom{w}{1}, \ldots, \binom{w}{p^{n} - 1}$) and obtain $\prod_{i=0}^{n-1} \binom{w_i}{k_i} |_\rho$?
We get exactly the term $\prod_{i=n-\ell}^{n-1} \binom{p-1}{k_i} \prod_{i=0}^{n-\ell-1} \binom{w_i}{k_i}$, where the constant factor $\prod_{i=n-\ell}^{n-1} \binom{p-1}{k_i}$ is always nonzero.
Thus, this linear map maps every term to (nonzero multiple of) a term. If all terms corresponding to $a, a+q, a+2q, \ldots$ map to distinct terms,
it implies this map is injective on the domain of all such $f$, and thus that the image of a nonzero $f$ is a nonzero polynomial as desired.

When do the terms corresponding to two multiples $a+k_1 q, a+k_2 q$ of $q$ map to the same term under this restriction? As we saw, this happens if and
only if they agree on their $n - \ell$ lowest order bits, which happens if and only if $2^{n-\ell}$ divides $(k_1 - k_2)q$. Since $q$ is coprime to $p$, this implies
$2^{n-\ell}$ divides $k_1 - k_2$. But $k_1 - k_2 \leq N/q$. Thus $\ell < \log_p q$ implies this can only happen if $k_1 - k_2 = 0$, that map is injective.
Therefore, we take $\ell = \lfloor \log_p(q-1) \rfloor$, which is the maximal integer less than $\log_p q$, and our conclusion follows.


\end{proof}

As a consequence of the above lemma and Proposition \ref{prop:reflex}, we can lower bound the degree of a symmetric nonzero function in the idea $\langle \chi_q \rangle$, when $n+1$ is a power of $p$.

\begin{corollary} Let $n > 1$ be an integer such that $n+1$ is a power of $p$. Let $f \in \langle \chi_q \rangle$ be a nonzero symmetric function,
then
$$
\deg(f) \ge n ( 1 - \frac{1}{p^\ell} ),
$$
where $\ell = \lfloor \log_p (q-1) \rfloor$.
\end{corollary}

In the case that $n+1$ is a power of $p$, and $q+1$ is a power of $p$, the above corollary gives lower bound $n ( 1 - \frac{1}{q-1})$, which is close to the optimal. Because
if we view $\chi_q$ as a symmetric function from $\{0,1,\ldots, n\}$ to $F$, it takes zero on $n - \lfloor n/q \rfloor - 1$ points, which implies there must exsit a nonzero symmetric
function in $\langle \chi_q \rangle$ of degree $n - \lfloor n/q \rfloor - 1$ by solving $n - \lfloor n/q \rfloor - 1$ in $n - \lfloor n/q \rfloor$ variables in the form
\eqref{equ:sym_basis}.

In the case that $n+1$ is not a power of $p$, we can reduce to the former case by applying a restriction $\rho$ with support size $n'$ such that $n-n'+1$ is a power of $p$.
However, we may lose a lot if $n+1$ is much above a power of $p$.

\begin{corollary} Let $n > 1$ be an integer, and $n' = n + 1 - p^{\lfloor \log_p (n+1) \rfloor}$ and thus $n-n'+1$ is a power of $p$. Let $f \in \langle \chi_q \rangle$ be a nonzero symmetric function,
then
$$
\deg(f) \ge (n-n') ( 1 - \frac{1}{p^\ell} ),
$$
where $\ell = \lfloor \log_p (q-1) \rfloor$.
\end{corollary}
\begin{proof}
 Let $f \in \langle \chi_q \rangle$ be a nonzero symmetric function with minimum degree. Let $\rho$ be a restriction restrict $n'$ bits to constant, either $0$ or $1$, such
that $f|_\rho \not= 0$. It's easy to see such restriction exists, because if all restrictions of size $n'$ restricts $f$ to zero, then $f$ is a zero function. Moreover,
$f|_\rho$ is also symmetric, and in the ideal $\langle \chi_q|_{\rho} \rangle$. By Proposition \ref{prop:reflex} and Lemma \ref{lem:restrict}, we know that
$$
\deg(f|_\rho) \ge  (n-n') ( 1 - \frac{1}{p^\ell} ).
$$
The conclusion follows by observing $\deg(f|_\rho) \le \deg(f)$.
\end{proof}

\section{Immunity and Circuit Lower Bound}

The following is a classical result due to Razborov, which says functions computed by $\acc$ circuits correlates with low degree polynomials
over $F_p[x_1, \ldots, x_n]/(x_1^2 = x_1, \ldots, x_n^2 = x_n)$.

\begin{theorem}
 \label{thm:Raz}
 \cite{Raz87} Let $C$ be an $\acc$ circuit of size $S$ and depth $d$. For every $\ell > 0$, there is a polynomial $p(x)$ in $F_p[x_1, \ldots, x_n]/(x_1^2 = x_1, \ldots, x_n^2 = x_n)$ of degree at most $((p-1)\ell)^d$ such that
$$
\Pr_{x \in \{0,1\}^n} [C(x) \not= p(x)] \le \frac{S}{2^\ell}.
$$
\end{theorem}

Therefore, one approach to prove $\acc$ circuit lower bound is to prove correlation bound of low degree polynomials. In Smolensky's 1993 paper \cite{Smo93},
 he proved Hilbert function is an ``invariant'' for low degree polynomials.

\begin{definition}
Fix the field $F$. The Hilbert function $h_m(S)$, where $S \subseteq \{0,1\}^n$, is defined as the dimension of the following subspace
$$
\{f|_S : f \in F[x_1, \ldots, x_n]/(x_1^2 = x_1, \ldots, x_n^2 = x_n), \deg(f) \le m\}.
$$
\end{definition}

Smolensky proved that high Hilbert function implies correlation bound with low degree polynomials.


\begin{theorem}
 \label{thm:Smo}
 \cite{Smo93} The distance of $f$, where $S$ is the zero set of $f$, to any degree $d$ polynomials (all nonzero is viewed as $1$) is lower bounded by
$$2 h_m(S) - |S|,$$
where $m \le (n-d-1)/2$.
\end{theorem}

The following observation relates Hilbert function with immunity. By the definition of Hilbert function,
\begin{eqnarray*}
h_m(S) & = & \dim \{ f|_S : \deg(f) \le m \} \\
& = & \binom{n}{\le m} - \dim \{ f \in \langle S \rangle : \deg(f) \le m \},
\end{eqnarray*}
where $\langle S \rangle$ denotes the ideal of functions vanishing on $S$, and $\binom{n}{\le m} = \sum_{i \le m} \binom{n}{i}$. If the immunity of $S$ is greater than $m$, which means $\dim \{ f \in \langle S \rangle : \deg(f) \le m\} = 0$, and thus $h_m(S)$ achieves the maximal $\binom{n}{\le m}.$

\vspace{0.2cm}
Since the immunity of $\chi_q$ is lower bounded by $n /2$,
$
h_m(S) =\binom{n}{\le m}
$
for any $m = (n-d-1)/2 < n/2$, where $S$ is the zero set of $\chi_q$. For all $d = o(\sqrt{n})$, we have
\begin{eqnarray*}
2h_m(S) - |S| & = & 2\binom{n}{\le m} - 2^n(1-\frac{1}{q} + o(1)) \\
& = & 2^n (1 - o(1)) - 2^n(1- \frac{1}{q} + o(1)) \\
& = & \frac{2^n}{q} - o(2^n),
\end{eqnarray*}
By Theorem \ref{thm:Smo}, function $\chi_q$ is different  from  any degree $o(\sqrt{n})$ polynomials on at least $2^n (1/q - o(1))$ points. Taking $\ell = O(\log n)$
and $S = n^{O(1)}$
in Theorem \ref{thm:Raz}, thus $C(x)$ can be approximated by a $o(\sqrt{n})$ function with error $o(1)$. Combining these two facts implies any polynomial size $\acc$ circuit can only output the correct answer on at most $2^n (1 - 1/q + o(1))$ points,
and this was proved by Smolensky \cite{Smo93}.

Note that above argument works as long as Boolean function $f$ has immunity $\ge n/2 - o(\sqrt{n})$
and $|1_f| = \Omega(2^n)$, then $f$ has exponential $\acc$ circuit lower bound.

\vspace{0.2cm}
For another example, let's consider the $q$th residue character function, $\Lambda_q : \{0, 1\}^n \to \{0,1\}$ on  finite
field $F_{2^n}$. Fix a basis $b_1, \ldots, b_n$ of $F_{2^n}$ over $F_2$. Map $\phi : \{0,1\}^n \to F_{2^n}$ is defined as
$$
\phi(x) = \sum_{i = 1}^n x_i b_i \in F_{2^n}.
$$
Then $\Lambda_q(x) = 1$ if and only if there exists $y\in F_{2^n}$ such that $y^q = x$. Kopparty \cite{Kop11} proved exponential $\acct$
circuit lower bound of the  $q$th residue character over $F_{2^n}$. In fact, he proved something stronger, which is the lower bound of
 computing a
large power in $F_{2^n}$. Here, we present a simple proof  by immunity.

Carlet and Feng \cite{CF08} proves the quadratic residue function has one sided immunity not less than $n/2$, and their proof also works for $q$th residue character function. Since it's a nice and simple proof, we reproduce the proof here.

\begin{theorem}
\label{thm:imm_res}
 Assume $q$ divides $2^n - 1$.
The immunity of $\neg \Lambda_q(x)$ over $F_2$ is greater than $d$, as long as $\binom{n}{\le d} \le 2^n / q$.
\end{theorem}
\begin{proof} Let $f$ be a polynomial in $\langle \neg \Lambda_q(x) \rangle$ with degree $\le d$, and we shall prove
$f = 0$.

The trick is to view $f$ as a function $\tilde{f}$ from $F_{2^n} \to F_{2^n}$ by the natural map $\phi$, given the basis $b_1, \ldots, b_n$
of $F_{2^n}$ over $F_2$. Given $f : F_2^n \to F_2$, define $\tilde{f} : F_{2^n} \to F_{2^n}$ by
$$
\tilde{f}(x) = f(x_1, x_2, \ldots, x_n),
$$
where $x = x_1 b_1 + \ldots + x_n b_n$. It's easy to see any function from $F_{2^n} \to F_{2^n}$ can be written as a univariate
polynomial of degree less than $2^n$. Thus, write
\begin{eqnarray}
\tilde{f}(x) & = & \sum_{0 \le i \le 2^n - 1} c_i x^i \nonumber \\
& = & \sum_{0 \le i \le 2^n-1} c_i \left(\sum_{j = 1}^n b_i x_i \right)^{\sum_{s = 0}^{n-1}i_s 2^s} \nonumber \\
& = & \sum_{0 \le i \le 2^n-1} c_i \prod_{s=0}^{n-1}\left(\sum_{j = 1}^n b_i^{i_s 2^s} x_i^{i_s} \right), \label{equ:F2n}
\end{eqnarray}
where $i = \sum_s i_s 2^s$ is the binary representation of $i$. Imagining \eqref{equ:F2n} is expanded, it's easy to see the coefficients of
$\prod_{i \in S} x_i$ for any $S \subseteq [n]$ should be in $\{0, 1\}$, and coincides with the expansion of $f: F_2^n \to F_2$, for they
are taking the same value on every $x_1, \ldots, x_n$.
 From this, we see the degree $f : F_2^n \to F_2$
is
$$
\max\{ w_2(i) : c_i \not= 0\},
$$
where $w_2(i)$ is defined as the number of $1$'s in the binary representation of $i$. Hence, assume
$$
\tilde{f}(x) = \sum_{0 \le i \le 2^n - 1\atop w_2(i) \le d} c_i x^i,
$$
and we will show $\tilde{f}(x) = 0$, that is, $c_i = 0$ for all $i$.

Let $\xi$ be a primitive root of $F_{2^n}$. Since $f$ is in $\langle \neg \Lambda_q(x) \rangle$, $\tilde{f}$
has to take $0$ on $\xi^0, \xi^q, \xi^{2q}, \ldots, \xi^{2^n - 1}$, that is,
\begin{equation}
\label{equ:lin}
\begin{pmatrix}
\xi^{0} & \xi^{0} & \cdots & \xi^{0} \\
\xi^{q i_1} & \xi^{q i_2} & \cdots & \xi^{q i_m} \\
\ldots & \ldots & \ddots & \ldots\\
\xi^{tq i_1} & \xi^{tq i_2} & \cdots & \xi^{tq i_m} \\
\end{pmatrix}
\begin{pmatrix}
c_{i_1}\\
c_{i_2}\\
\ldots\\
c_{i_m}
\end{pmatrix} = 0,
\end{equation}
where $t = (2^n - 1) / q$ and $i_1, \ldots, i_m$ enumerates all $i$ with $w_2(i) \le d$. By assumption $\binom{n}{\le d} \le 2^n / q$, we have $m \le t$.
Since the matrix on the left hand side of \eqref{equ:lin} has full rank $m$ by Vandermonde determinant formula,
$c_{i_1} = c_{i_2} = \ldots = c_{i_m} = 0$, which completes the proof.
\end{proof}

Let $S$ be the one-set of $\Lambda_q$. For integer $m$ such that $\binom{n}{\le m} \ge 2^n / q$, by the above
theorem, we have
$$
2h_m(S) - |S| \ge 2h_{m'}(S) - |S|  \ge 2^n (\frac{1}{q} - o(1)),
$$
where $m'$ is the largest integer such that $\binom{n}{\le m'} \le 2^n / q$, and thus $m' = n/2 - \Theta(\sqrt{n})$ for fixed $q$. Combining with Theorem \ref{thm:Smo}, function $\Lambda_q$ is different  from  any degree $o(\sqrt{n})$ polynomials at $2^n (1/q - o(1))$ points. Following the same argument as we did for MOD function, any polynomial size
$\acct$ circuit can agree with $\Lambda_q$ on at most $2^n (1 - 1/q + o(1))$ points.

Moreover, by the immunity argument, we can prove the following result, which improves the size bound by Kopparty \cite{Kop11} from $2^{n^{1/(20d)}}$ to $2^{n^{1/({(2+\eps)}d})}$, where $\eps > 0$ is arbitrarily small, where the constant $1/(2d)$ on the double exponent seems to be the best we can
get by the direct Razborov-Smolensky approach.

\begin{theorem} For every $\acct$ circuit $C : \{0,1\}^n \to \{0,1\}$ of depth $d$ and size $S \le 2^{n^{1/{(2+\eps)}d}}$,
where $\eps > 0$ is arbitrarily small, we have
\begin{equation}
\label{equ:qthres}
\Pr_x[C(x) = \Lambda_q(x)] \le 1 - \frac{1}{q} + o_n(1),
\end{equation}
where $o_n(1)$ goes to $0$ as $n$ goes to infinity after $q$ and $\eps$ are fixed.
\end{theorem}
\begin{proof} Applying Razborov's Theorem \ref{thm:Raz} by taking $\ell = n^{1/(2 + 0.5\eps)d}$, there exists a polynomial of degree
$\le \ell^d = n^{1/(2+0.5\eps)}$, such that,
$$
\Pr_x[C(x) \not= g(x)] \le \frac{S}{2^{n^{1/(2+0.5\eps)d}}} = o_n(1).
$$
Meanwhile, by Theorem \ref{thm:Smo} and Theorem \ref{thm:imm_res},
\begin{eqnarray*}
\Pr_x[\Lambda_q(x)  \not=  g(x)] & \ge & (2h_m(S) - |S|)/2^n \\
& \ge & (2h_{(n-\ell^d - 1)/2}(S) - |S|)/2^n \\
& \ge & (2h_{n/2-o(\sqrt{n})}(S) - |S|)/2^n \\
& = & \frac{1}{q} - o_n(1).
\end{eqnarray*}
By triangle inequality,
$$
\Pr[\Lambda_q(x) \not= C(x)] \ge \Pr[\Lambda_q(x) \not= g(x)] - \Pr[C(x) \not= g(x)] = \frac{1}{q} - o_n(1),
$$
which proves the theorem.
\end{proof}

In fact, what Kopparty proved in \cite{Kop11} is for $q$th residue function $\Lambda_q : F_{2^n} \to \{0, 1, \ldots, q-1\}$ instead of the
$q$th residue character function.
We can easily modify the above argument for $q$th residue function as follows, where
the right hand side of \eqref{equ:qthres} will become $1/q + o(1)$. Given $\eps > 0$, suppose for contradiction that there exists a circuit
$C$ of depth $d$ and size $2^{n^{1/{(2+\eps)}d}}$ agrees with $\Lambda_q$ on  $\ge 1/q'$ fractions, where $1/q' > 1/q$. Again, taking $\ell = n^{1/(2 + 0.5\eps)d}$ in Theorem \ref{thm:Raz}, there exist polynomials $g_0, \ldots, g_{q-1}$ of degree $\le \ell^d = n^{1/(2+0.5\eps)} = o(\sqrt{n})$ which agrees with $P_0, \ldots, P_{q-1}$ on $1-o(1)$ fraction respectively, which implies,
$$
\sum_i \Pr_{x} [g_i(x) = \mathbbm{1}_{P_i}(x)] \ge 1/q' - o(1),
$$
where $P_i = \{ x\in F_{2^n} : \Lambda_q(x) = i\}$.
Denote by $S = \{ x : g_i(x) = 1 \text{ for some } i \in P_i \}$, where $S \ge (1/q' - o(1))2^n$. By the Hilbert function and immunity argument, all polynomials of degree $\le d$, where $\binom{n}{\le d} \ge 2^n (1/q + o(1))$, can represent any
function restricting on $P_i$. Since the existence of $g_0, g_1, \ldots, g_{q-1}$, degree $d + \max_i \deg(g_i)$ polynomials are sufficient to represent any functions
on $S$. The contradiction comes from a double counting: the number of such polynomials is upper bounded by $2^{\binom{n}{\le d + \deg(g)}} =
2^{2^n (1/q + o(1))}$, while the number of Boolean functions on $S$ is $2^{|S|} \ge 2^{2^n (1 / q' - o(1))}$, where $1/q'$ is strictly greater than
$1/q$.

\section{Conclusion and Open Problems}

In this paper, we prove tight lower bounds on the smallest degree of a nonzero polynomial in the ideal generated by $MOD_q$ or $\neg MOD_q$ in the polynomial ring $F_p[x_1, \ldots, x_n]/(x_1^2 = x_1, \ldots, x_n^2 = x_n)$,
$p,q$ are coprime.
For the $MOD_q$, our lower bound $n/2$ can be achieved when $n$ is a multiple of $2q$; For $\neg MOD_q$, our lower bound
 $\lfloor \frac{n+q-1}{q} \rfloor$  is exact for every $n$ and $q$, independent of prime $p$.
The previous best results $\lfloor \frac{n}{2(q-1)} \rfloor$ is by Green \cite{Gre00}, which uses different techniques.

For the immunity of $\neg MOD_q$, our lower bound is exact; for the immunity of $MOD_q$, our lower bound $n/2$ is tight
for those $n$ which is a multiple of $2q$; for other cases, there is a gap of size at most $q$ (Experiment shows
the gap is at most $1$).
It would be nice if this small gap can be closed.

\begin{question}
What is the exact immunity of $\chi_q$ over filed $F_p$?
\end{question}

In Section 3, after proving the lower bound of the immunity of $\chi_q$, we also constructed functions in
$\langle \chi_q \rangle$ with matching or nearly matching lower bound. It is natural to ask the following question.

\begin{question}
Characterize all the nonzero functions in $\langle \chi_q \rangle$ or $\langle \neg \chi_q \rangle$ with the minimum possible
degree.
\end{question}

In Section 7, we observe that if a Boolean function has immunity $\ge n/2 - o(\sqrt{n})$ and $|1_f| = \Omega(2^n)$, then $f$ is uncorrelated with low  degree
 polynomial in ring $R = F_p[x_1, \ldots, x_n]/(x_1^2=x_1, \ldots, x_n^2 = x_n)$, and thus implies exponential $\acc$ circuit lower bound.
We feel some complexity measure of a Boolean function might be closely related to some nice algebraic properties
of the ideal $\langle f \rangle$, like immunity or Gr\"{o}bner basis. For a random Boolean function, such properties are difficult
to compute. However, for some natural functions we are interested in, like Clique, Mod and Permanent, such algebraic properties might be exceptional and possible to compute. It's likely that there are more connections between nice properties of the ideal $\langle f \rangle$
in $R$ and some circuit complexity measures.
In a recent paper \cite{KS12}, Kopparty and Srinivasan proved that $\Omega(n)$ lower bound
of two-sided immunity over $F_2$ implies superlinear $\acct$ circuit lower bound, and
$$
\frac{n}{2} - \frac{n}{(\log n)^{\omega(1)}}
$$
lower bound of two-sided immunity implies superpolynomial $\acct$ circuit lower bound.
 Therefore, we have the following general
open question.

\begin{question}
Are there more connections between circuit complexity of Boolean function $f$ and some algebraic properties of ideal $\langle f \rangle$
in the ring $F_p[x_1, \ldots, x_n]/(x_1^2=x_1, \ldots, x_n^2 = x_n)$?
\end{question}

\section*{Acknowledgment}
Yuan Li would like to thank Sasha Razborov for illuminating discussions.

\end{document}